\documentclass[runningheads,a4paper]{llncs}

\usepackage{graphicx}

\usepackage{amsmath, amssymb}
\usepackage[english]{babel}
\usepackage{xspace, url, verbatim, array, float, color}
\usepackage[mathscr]{eucal}
\usepackage{cite}
\usepackage{enumerate}
\usepackage{calc}

\def\cC{\mathcal{C}} \def\cP{\mathcal{P}}

\newcommand\abs[1]{\left\lvert#1\right\rvert}

\newcommand\Rset{\mathbb{R}}

\def\phi{\varphi}

\newcommand\dRF{\ensuremath{d_{\text{RF}}}\xspace}
\newcommand\dGRF{\ensuremath{d_{\text{GRF}}}\xspace}
\newcommand\dJRFk{\ensuremath{d_{\text{JRF}}^{(k)}}\xspace}
\newcommand\dJRFone{\ensuremath{d_{\text{JRF}}^{(1)}}\xspace}

\newcommand\etal{\emph{et~al.\@}\xspace}

\DeclareMathOperator{\symmdiff}{\triangle}

\begin{document}

\title{The generalized Robinson-Foulds metric}

\author{Sebastian B\"ocker\thanks{Equal contribution.} \and Stefan
  Canzar$^{\star}$ \and Gunnar W.\ Klau$^{\star}$}

\institute{Chair for Bioinformatics, Friedrich Schiller University Jena,
  Germany, \email{sebastian.boecker@uni-jena.de} \and 
  Center for Computational Biology, McKusick-Nathans Institute of Genetic
  Medicine, Johns Hopkins University School of Medicine, Baltimore, Maryland,
  USA, \email{canzar@jhu.edu} \and Life Sciences Group, Centrum Wiskunde \&
  Informatica, Amsterdam, The Netherlands, \email{gunnar.klau@cwi.nl}}

\date{\today}

\maketitle


\begin{abstract}
  The Robinson-Foulds (RF) metric is arguably the most widely used measure of
  phylogenetic tree similarity, despite its well-known shortcomings: For
  example, moving a single taxon in a tree can result in a tree that has
  maximum distance to the original one; but the two trees are identical if we
  remove the single taxon. To this end, we propose a natural extension of
  the RF metric that does not simply count \emph{identical} clades but
  instead, also takes \emph{similar} clades into consideration. In contrast to
  previous approaches, our model requires the matching between clades to
  respect the structure of the two trees, a property that the classical RF
  metric exhibits, too. We show that computing this generalized RF metric is,
  unfortunately, NP-hard.  We then present a simple Integer Linear Program
  for its computation, and evaluate it by an all-against-all comparison of 100
  trees from a benchmark data set.  We find that matchings that respect the
  tree structure differ significantly from those that do not, underlining
  the importance of this natural condition.
\end{abstract}


\section{Introduction}

In 1981, Robinson and Foulds introduced an intriguingly simple yet
intuitively well-motivated metric, which is nowadays known as
\emph{Robinson-Foulds (RF) metric}~\cite{robinson81comparison}.  Given two
phylogenetic trees, this metric counts the number of splits or clades induced
by one of the trees but not the other.  The RF metric is highly conservative,
as only perfectly conserved splits or clades do not count towards the
distance.  The degree of conservation between any pair of clades that is not
perfectly conserved, does not change the RF distance.  See
Fig.~\ref{fig:two-trees} for an example of two trees that are structurally
similar but have maximum RF distance.

Other measures for comparing phylogenetic trees do capture that the trees in
Fig.~\ref{fig:two-trees} are structurally similar: The Maximum Agreement
Subtree (MAST) score~\cite{finden85obtaining, kao01even} of the two trees is
$9$, where $10$ is the highest possible score of two trees with $10$ leaves.
Secondly, the triplet distance counts the number of induced triplet trees on
three taxa that are not shared by the two trees~\cite{bansal11comparing,
  critchlow96triples}.  Both measures are less frequently applied than the RF
metric, and one may argue that this is due to certain ``issues'' of these
measures: For example, if the trees contain (soft) polytomies or arbitrarily
resolved polytomies, then we may have to exclude large parts of the trees
from the MAST due to a single polytomy.  Lastly, there are
distance measures based on the number of branch-swapping operations to
transform one tree into another; many of these measures are computationally
hard to compute~\cite{allen01subtree}.  Such tree modifications are routinely
used in local search optimization procedures, but rarely to compute distances
in practice.

\begin{figure}[tb]
\centering
\includegraphics[width=0.9\textwidth]{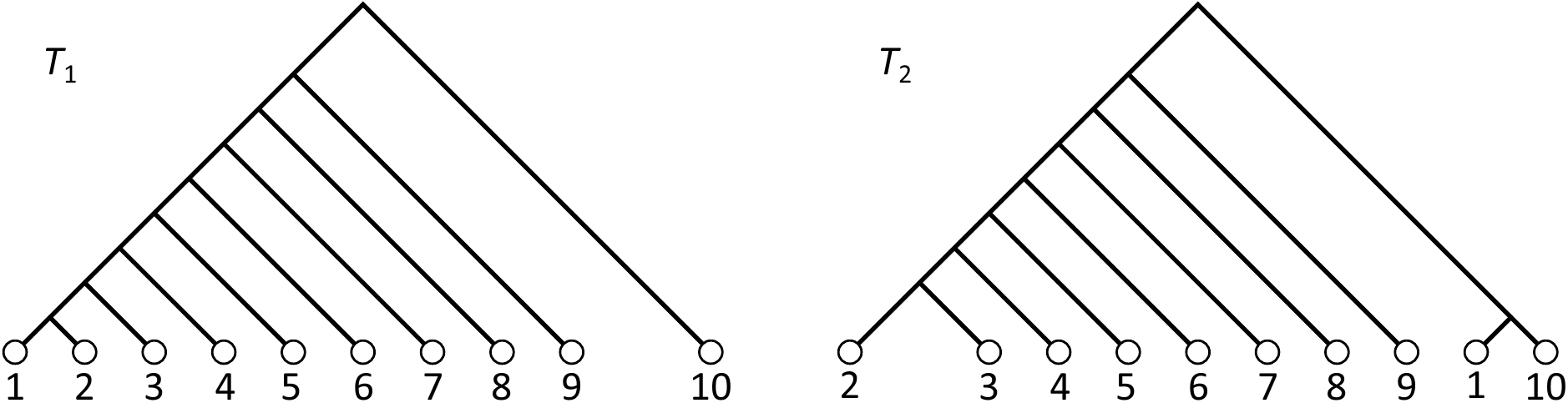}
\caption{Two rooted phylogenetic trees.  Despite their high similarity, the
  RF distance of these two trees is $16$, the maximum distance of two rooted
  trees with ten leaves.}
\label{fig:two-trees}
\end{figure}

From an applied view, the comparison of two phylogenetic trees with identical
taxa set has been frequently addressed in the literature~\cite{nye06novel,
  munzner03treejuxtaposer, griebel08epos}.  This is of interest for comparing
phylogenetic trees computed using different methods, output trees of an
(MC)MCMC method, or host-parasite comparisons.  Mutzner
\etal~\cite{munzner03treejuxtaposer} introduced the ``best corresponding
node'' concept which, unfortunately, is not symmetric: Node $a$ in the first
tree may correspond to node $b$ in the second, whereas $b$ corresponds to a
different node $c$ in the first tree, and so on.  Nye \etal~\cite{nye06novel}
suggested to compute a matching between the inner nodes of the two trees,
thereby enforcing symmetry.  Later, Bogdanowicz~\cite{bogdanowicz08comparing}
and, independently, Lin \etal~\cite{YuLin:2012dj} proposed to use these
matchings to introduce a ``generalized'' version of the RF distance, see
also~\cite{bogdanowicz12matching}.  Using matchings for comparing trees as
part of MAST computations, was pioneered by Kao \etal~\cite{kao01even}.

Here, we present a straightforward generalization of the RF distance that
allows us to relax its highly conservative behavior.  At the same time, we
can make this distance ``arbitrarily similar'' to the original RF distance.
Unfortunately, computing this new distance is NP-hard, as we will show in
Section~\ref{sec:complexity}.  Our work generalizes and formalizes that of
Nye \etal~\cite{nye06novel}: Their clade matching does not respect the
structure of the two trees, see Fig.~\ref{fig:two-trees} and below.  As a
consequence, the matching distances from~\cite{bogdanowicz08comparing,
  YuLin:2012dj} are no proper generalization of the RF distance: These
distances treat the two input trees as collections of (unrelated) clades but
\emph{ignore the tree topologies}.  In contrast, the RF distance does respect
tree topologies, and so does our generalization.

In the following, we will concentrate on rooted phylogenetic trees.


\section{The Generalized Robinson-Foulds distance}

Let $T = (V,E)$ be a rooted phylogenetic tree over the set of taxa~$X$: That
is, the leaves of $T$ are (labeled~by) the taxa~$X$.  We assume that $T$ is
arboreal, so all edges of $T$ are pointing away from the root.  In the
following, we assume that any tree is an arboreal, rooted phylogenetic tree,
unless stated otherwise.  A set $Y \subseteq X$ is a \emph{clade} of $T$ if
there exists some vertex $v \in V$ such that $Y$ is the set of leaves
below~$v$.  We call $Y$ \emph{trivial} if $\abs{Y} = 1$ or $Y = X$.  Since
the trivial clades are identical for any two trees with taxa set~$X$, we will
restrict ourselves to the set $\cC(T)$ of non-trivial clades of~$T$.  Let
$\cP(X)$ be the set of subsets of~$X$.

Let $T_1,T_2$ be two phylogenetic trees over the set of taxa~$X$, and let
$\cC_j := \cC(T_j)$ for $j=1,2$ be the corresponding sets of non-trivial
clades.  The original RF distance counts zero whenever we can find a clade in
both trees, and one if we find it in exactly one tree.  We want to relax this
by computing a matching between the clades of the two trees, and by assigning
a cost function that measures the dissimilarity between the matched
clades.
To this end, we define a \emph{cost function}
\begin{equation} \label{equ:delta}
  \delta: \bigl( \cP(X) \cup \{-\} \bigr) \times \bigl( \cP(X) \cup \{-\}
  \bigr) \to \Rset_{\ge 0} \cup \{\infty\}\enspace.
\end{equation}
Now, $\delta(Y_1,Y_2)$ measures the dissimilarity of two arbitrary clades
$Y_1,Y_2 \subseteq X$.  The symbol `$-$' is the gap symbol, and we define
$\delta(Y_1,-) > 0$ to be the cost of leaving some clade $Y_1$ of the first
tree without a counterpart in the second tree; analogously, we
define~$\delta(-,Y_2) > 0$.

\subsection{Matchings and arboreal matchings}\label{sec:arboreal_match}

Let $m \subseteq \cC_1 \times \cC_2$ be a \emph{matching} between $\cC_1$
and~$\cC_2$: That is, $(Y_1,Y_2),(Y_1',Y_2) \in m$ implies $Y_1 = Y_1'$, and
$(Y_1,Y_2),(Y_1,Y_2') \in m$ implies $Y_2 = Y_2'$.  We say that $Y_1 \in
\cC_1$ (or $Y_2 \in \cC_2$) is \emph{unmatched} if there is no $(Y_1',Y_2')
\in m$ with $Y_1 = Y_1'$ (or $Y_2 = Y_2'$, respectively).  We define the
\emph{cost} $d(m)$ of the matching $m$ as:
\begin{equation} \label{equ:cost-matching}
  d(m) := \sum_{(Y_1,Y_2) \in m} \delta(Y_1,Y_2)
  + \sum_{\substack{Y_1 \in \cC_1 \\ \text{$Y_1$ unmatched}}} \!\!
  \delta(Y_1,-)
  + \sum_{\substack{Y_2 \in \cC_2 \\ \text{$Y_2$ unmatched}}} \!\!
  \delta(-,Y_2)
\end{equation}

Now, we could define a generalization of the Robinson-Foulds distance between
$T_1,T_2$ (with respect to~$\delta$) to be the minimum cost of any matching
between $\cC_1$ and~$\cC_2$.  One can easily see that for $\delta(Y,Y)=0$,
$\delta(Y,Y') = \infty$ for $Y \ne Y'$, and $\delta(Y,-) = \delta(-,Y) = 1$
we reach the original RF distance.

How can we compute a matching of minimum cost?  This is actually
straightforward: We define a complete bipartite graph $G$ with vertex set
$\cC_1 \cup \cC_2$, and for any pair $C_1 \in \cC_1$, $C_2 \in \cC_2$ we
define the weight of the edge $(C_1,C_2)$ as $w(C_1,C_2) := \delta(C_1,-) +
\delta (-,C_2) - \delta(C_1,C_2)$.  Now, finding a matching with minimum cost
corresponds to finding a maximum matching in~$G$.  In case $\delta$ is a
metric, all edges in $G$ have non-negative weight.

Unfortunately, finding a minimum cost matching will usually result in an
unexpected---and undesired---behavior: Consider the two trees from
Fig.~\ref{fig:two-trees} together with the cost function
\begin{equation}\label{eq:exmpl_metric}
  \delta(Y_1,Y_2) = \abs{Y_1 \cup Y_2} - \abs{Y_1 \cap Y_2} = \abs{Y_1
    \symmdiff Y_2}
\end{equation}
which is the cardinality of the \emph{symmetric difference} $Y_1 \symmdiff
Y_2$ of $Y_1,Y_2$.  In addition, we define $\delta(Y,-) = \delta(-,Y) =
\abs{Y}$.  We note that $\delta$ is a metric.  One can easily see that the
matching with minimum cost matches clade $\{1,\dots,j\}$ from $T_1$ to
$\{2,\dots,j\}$ from $T_2$ for all $j=3,\dots,10$.  But in addition, clade
$\{1,2\}$ from $T_1$ is matched to clade $\{1,10\}$ from $T_2$, since
\[
  \delta \bigl( \{1,2\}, \{1,10\} \bigr) = 2 < 4 = \delta \bigl(\{1,2\},-
  \bigr) + \delta \bigl(-,\{1,10\} \bigr)\enspace.
\]
This means that the matching with minimum cost does not respect the structure
of the two trees $T_1,T_2$: Clade $\{1,2\}$ in $T_1$ is a subclade of all
$\{1,\dots,j\}$ whereas clade $\{1,10\}$ in $T_2$ is no subclade of any
$\{2,\dots,j\}$, for $j=3,\dots,10$.  To this end, clades $\{1,2\}$ and
$\{1,10\}$ should not be matched in a ``reasonable'' matching.

We say that a matching $m$ is \emph{arboreal} if no pair of matched
clades is in \emph{conflict}, that is, for any $(Y_1,Y_2), (Y_1',Y_2') \in m$, one of the three cases holds: 
\begin{enumerate}[(i)]
\item $Y_1 \subseteq Y_1'$ and $Y_2 \subseteq Y_2'$;

\item $Y_1 \supseteq Y_1'$ and $Y_2 \supseteq Y_2'$; or 

\item $Y_1 \cap Y_1' = \emptyset$ and  $Y_2 \cap Y_2' = \emptyset$. 
\end{enumerate}
This allows us to define the \emph{generalized Robinson-Foulds distance}
between $T_1,T_2$ (with respect to~$\delta$) to be the minimum cost of a
arboreal matching between $\cC_1$ and~$\cC_2$. Whereas it is straightforward
to compute a bipartite matching of minimum cost, it is less clear
how to obtain an minimum cost arboreal bipartite matching.  The formal
problem statement is as follows:

\medskip

\noindent\textbf{Minimum Cost Arboreal Bipartite Matching.}  Given two rooted
phylogenetic trees $T_1,T_2$ on $X$ and a cost function $\delta$, find a
arboreal matching between $\cC(T_1)$ and $\cC(T_2)$ of minimum cost, as
defined in~\eqref{equ:cost-matching}.

\medskip

This problem differs from the NP-complete \emph{tree-constrained
    bipartite matching problem introduced in \cite{tcbm} in that cases (i)
  and (ii) are considered infeasible in \cite{tcbm}.} Unfortunately, the
problem remains NP-complete, as we will show in Sec.~\ref{sec:complexity}.

\medskip

For arbitrary cost functions $\delta$ we cannot draw conclusions about the
resulting generalized Robinson-Foulds distance.  But in case $\delta$ is a
metric, this distance is a metric, too:

\begin{lemma} \label{lem:metric}
Given a metric $\delta$ as defined in \eqref{equ:delta}; then, the induced
generalized Robinson-Foulds distance $\dGRF$ is a metric on the set of
phylogenetic rooted trees on~$X$.
\end{lemma}

For the proof, the central point is that the combination of two arboreal
matchings is also a arboreal matching; we defer the details to the full
version of this paper.

\subsection{The Jaccard-Robinson-Foulds metric}

Up to this point, we have assumed that $\delta$ can be an arbitrary metric.
Now, we suggest one particular type that, again, appears quite naturally as a
generalization of the original Robinson-Foulds metric: Namely, we will
concentrate on a measure that is motivated by the Jaccard index $J(A,B) =
\abs{A \cap B} / \abs{A \cup B}$ of two sets $A,B$.  For two clades $Y$,
$Y'$, we define the \emph{Jaccard weights of order $k$} as
\begin{equation} \label{equ:jrf}
   \delta_k(Y,Y') := 2 - 2 \cdot \left( \frac{\abs{Y \cap Y'}}{\abs{Y \cup
       Y'}} \right)^k
\end{equation}
where $k \ge 1$ is an arbitrary (usually integer) constant.  In addition, we
define $\delta_k(Y,-) = \delta_k(-,Y') = 1$ and, for completeness,
$\delta_k(\emptyset,\emptyset) = 0$.  The factor ``2'' in eq.~\eqref{equ:jrf}
is chosen to guarantee compatibility with the original Robinson-Foulds
metric.  Nye \etal~\cite{nye06novel} suggested a similar metric without the
exponent~$k$.  It is straightforward to check that \eqref{equ:jrf} defines a
metric, see \cite{deza97geometry} and the full version of this paper.  We
call the generalized Robinson-Foulds metric using $\delta_k$ from
\eqref{equ:jrf} the \emph{Jaccard-Robinson-Foulds} (JRF) metric of
\emph{order}~$k$, and denote it by $\dJRFk$. More precisely, for two trees
$T_1$, $T_2$, $\dJRFk(T_1,T_2)$ denotes the minimum cost of any matching
between $\mathcal{C}(T_1)$ and $\mathcal{C}(T_2)$, using $\delta_k$ from
\eqref{equ:jrf} in \eqref{equ:cost-matching}.

For any two trees and any $k \ge 1$ we clearly have $\dJRFk(T_1,T_2) \le
\dRF(T_1,T_2)$, as the matching of the RF metric is clearly arboreal.  For $k
\to \infty$ we reach $\delta_k(Y,Y) \to 0$ and $\delta_k(Y,Y') \to 1$ for $Y
\ne Y'$, the inverse Kronecker delta.  To this end, the JRF metric $\dJRFk$
also converges to the original Robinson-Foulds metric~$\dRF$.  Furthermore,
for any two trees $T_1,T_2$ there exists some $k'$ such that for all $k \ge
k'$, the matchings for $\dRF$ and $\dJRFk$ are ``basically identical'': All
exact clade matches will be contained in the matching of $\dJRFk$.  We defer
the details to the full version of this paper.


\section{Complexity of the problem} \label{sec:complexity}

In this section we prove hardness of the minimum arboreal matching problem,
even if $\delta$ (and thus the induced RF distance, see
Lemma~\ref{lem:metric}) is a metric.

In the following we devise a polynomial-time reduction $\tau$ from
$(3,4)$-SAT, the problem of deciding whether a Boolean formula in which every
clause is a disjunction of exactly 3 literals and ever variable occurs 4
times, has a satisfying assignment. This problem was shown to be NP-hard in
\cite{journals/dam/Dubois90}. Given a formula $\phi$ with $m$ clauses over
$n$ variables, we construct a minimum arboreal matching instance $I$ under
metric \eqref{eq:exmpl_metric}, such that $\phi$ is satisfiable if and only
if $I$ admits a matching of cost $d(M_0)-10n-26m - 5\cdot 2^4 - (k+1)2^k -
q$, where $M_0$ is the empty matching.

\begin{figure}[tb]
   \centering
   \includegraphics[width=0.7\linewidth]{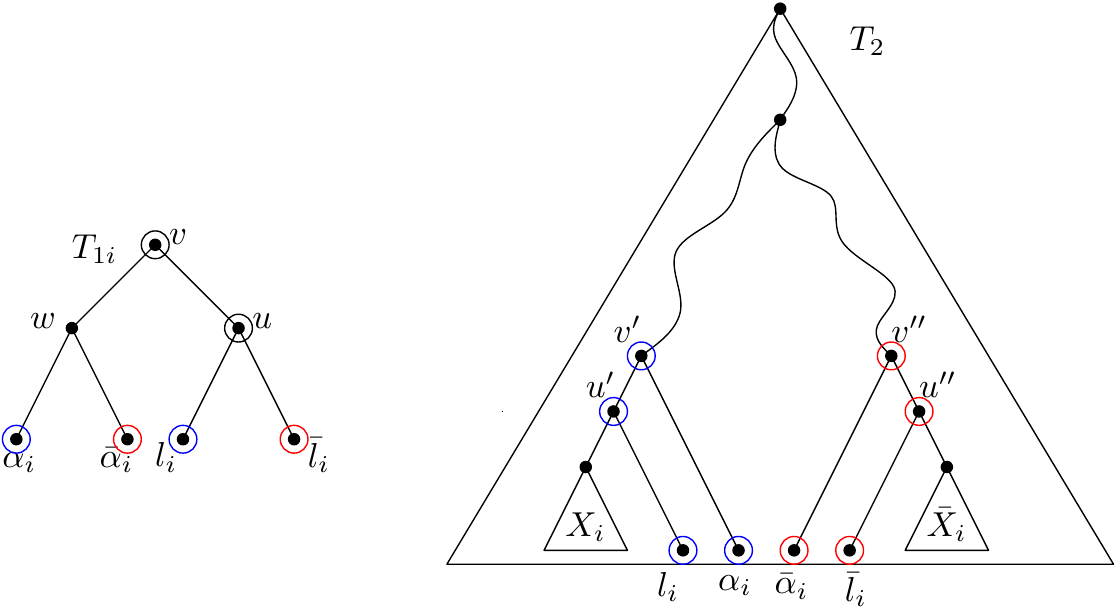}
   \caption{Variable Gadget. Vertices covered by optimal matchings $M_l$ and
     $M_{\bar{l}}$ are marked in blue and red, respectively. Vertices marked
     in black are covered in both optimal matchings.}
   \label{fig:var_gadget}
\end{figure}

For each variable $x_i$ we construct a gadget as shown in
Figure~\ref{fig:var_gadget}.  The next lemma shows that, under certain
assumptions, there are precisely two optimal solutions to the variable
gadgets.  We will use these two matchings to represent a truth assignment to
variable $x_i$.

\begin{lemma} \label{theo:var_gadget}
 Consider the gadget of a variable $x_i$ as depicted in
 Figure~\ref{fig:var_gadget}.  Under the restriction that none of the
 ancestors of nodes $v$ and $v'$ is matched, there are two optimal
 matchings of trees $T_{1i}$ and $T_{2}$ of cost $d(M_0)-10$.  $M_l$ contains
 $(v,v')$ and $(u,u')$ and matches leaves labeled $l_i$ and $\alpha_i$, and
 $M_{\bar{l}}$ contains $(v,v'')$ and $(u,u'')$ and matches leaves labeled
 $\bar{l}_i$ and $\bar{\alpha}_i$.
\end{lemma}

\begin{proof}
In the following was assume that none of the ancestors of $v$ and $v'$ can be
matched.  Let $M_0=\emptyset$ be the empty matching between $T_{1i}$ and
$T_{2}$, and let $M_a$ denote the matching that matches all leaves with
identical labels. Then, $M_a$ is maximal and $d(M_a)=d(M_0)-8$.  If we match
$u$ to either $u'$ or to $u''$, a feasible matching cannot match leaves
labeled $\bar{l}_i$ or leaves labeled $l_i$, respectively. Similarly,
matching $w$ to $v'$ or to $v''$ invalidates the matching of leaves labeled
$\bar{\alpha}_i$ or leaves labeled $\alpha_i$, respectively. In both cases
the overall cost remains unchanged compared to $M_a$.  If we match $v$ to
$v'$, only leaves labeled $l_i$ and $\alpha_i$ can be matched to
corresponding leaves in $T_{2}$. A feasible matching of node $w$ to any node
in $T_{2i}$ does not reduce the total cost, since none of the labels
of descendants of $v'$ contains $\alpha_i$ or
$\bar{\alpha}_i$. However, matching $u$ to $u'$ does not introduce any
conflict and further decreases the cost. The resulting matching (see
Figure~\ref{fig:var_gadget}), $M_l$, has cost $d(M_l)=d(M_0)-10$. By a
symmetric argument, a maximum matching $M_{\bar{l}}$ containing $(v,v'')$
matches $u$ to $u''$ and leaves labeled $\bar{l}_i$ and $\bar{\alpha}_i$,
with $d(M_{\bar{l}})=d(M_0)-10$.
\end{proof}

For each clause $C_j$ we construct a clause gadget as shown in
Figure~\ref{fig:clause_gadget}.

\begin{figure}[tb]
   \centering
   \includegraphics[width=\linewidth]{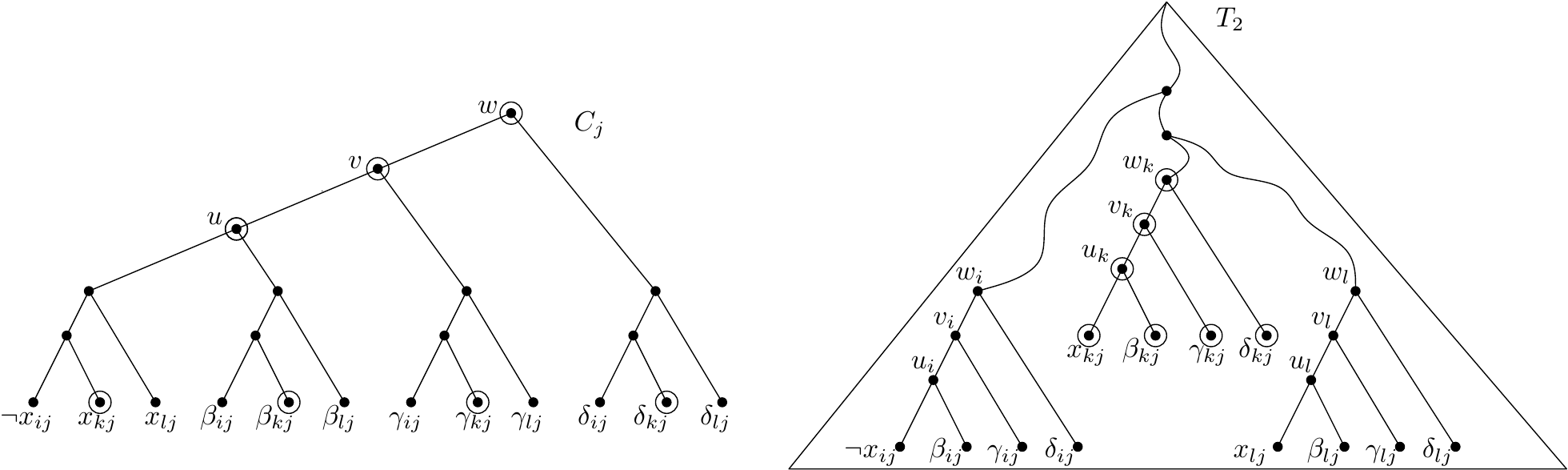}
   \caption{Clause gadget for clause $C_j=(\neg x_i \vee x_k \vee
     x_l)$. Vertices covered by an optimal matching are
     marked.}  \label{fig:clause_gadget}
\end{figure}

\begin{lemma}\label{theo:clause_gadget}
Consider the gadget of a clause $C_j$ as depicted in
Figure~\ref{fig:clause_gadget}.  Under the restriction that no common
ancestor of $w_i$, $w_k$, or $w_l$ is matched, there exists an optimal
matching $M$ of $C_j$ and $T_2$ that matches all vertices in one of the
subtrees rooted at $w_i$, $w_k$, or $w_l$ and none of the remaining vertices,
and has cost $d(M)=d(M_0)-26$.
\end{lemma}

\begin{proof}
Let $M_0=\emptyset$ be the empty matching between $C_j$ and $T_2$, and let
$M_a$ denote the matching that matches all leaves with identical labels. Then
$M_a$ is maximal and $d(M_a)=d(M_0)-24$.  Matching any non-leaf node below
$u$, $v$, or $w$ in $C_j$ to a node in $T_2$ that is not an ancestor of
$w_i$, $w_k$, or $w_l$, yields a matching of cost at least as high as
$d(M_a)$: At most one leaf in the subtree rooted at such a node $u'$ can be
matched to its corresponding leaf in $T_2$, while the label overlap of $u'$
with nodes in $T_2$ that are not ancestors of $w_i$, $w_k$, or $w_l$, is at
most $1$.

If node $u$ is matched to a node in $T_2$ with maximal label overlap that is
not an ancestor of $w_i$, $w_k$, or $w_l$, only $2$ leaves in the subtree
rooted at $u$ can be matched to the corresponding leaves in $T_2$. If the
remaining nodes in $T_1$ are matched according to $M_a$ the resulting
matching has cost $d(M_0)-20$.

Matching node $v$ to a node in $T_2$ with maximal overlap that is not an
ancestor of $w_i$, $w_k$, or $w_l$, allows only $3$ leaves in the subtree
rooted at $v$ to be matched to the corresponding leaves in
$T_2$. Additionally node $u$ can be matched to a node in $T_2$ with label
overlap of size $2$. Matching the remaining nodes in $C_j$ according to $M_a$
yields a matching of cost $d(M_0)-22$.

Finally, if node $w$ is matched to a node in $T_2$ with maximal label overlap
that is not an ancestor of $w_i$, $w_k$, or $w_l$, in total $4$ leaves in
$C_j$ can be matched to the corresponding leaves in $T_2$. At the same time,
$u$ and $v$ can be matched to nodes with maximal label overlap, yielding a
matching $M$ of cost $d(M)=d(M_0)-26$ (see Figure~\ref{fig:clause_gadget}).
Since all edges in $M$ have maximum label overlap
under the assumption that no common ancestor of $w_i$, $w_k$, or $w_l$ is
matched, $M$ is optimal.
\end{proof}

\begin{figure}[tb]
   \centering
   \includegraphics[width=0.3\linewidth]{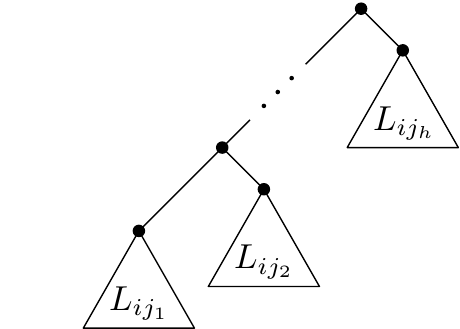}
   \caption{Module $X_i$ in the variable gadget for $x_i$ is composed of one
     tree $L_{ij}$ for each occurrence $j$ of positive literal
     $x_i$.}\label{fig:X}
 \end{figure}

Next, we show how variable and clause gadgets together form $\tau(\phi)$.
For each occurrence of a positive or negative literal $l_i$ or $\bar{l}_i$ in
a clause $j$ we denote the subtrees rooted at $w_i$, $w_k$, and $w_l$ in
$T_2$ (Figure~\ref{fig:clause_gadget}) by $L_{ij}$ or $\bar{L}_{ij}$,
respectively. $T_2$ in Figure~\ref{fig:clause_gadget} show trees
$\bar{L}_{ij}$, $L_{kj}$, and $L_{lj}$.  Let $j_1,\dots,j_h$ be the indices
of clauses in which positive literal $l_i$ occurs. Then, module $X_i$ in
Figure~\ref{fig:var_gadget} is constructed as shown in Figure~\ref{fig:X}.
Module $\bar{X}_i$ is analogously composed of trees~$\bar{L}$.

\begin{figure}[tb]
  \centering
  \includegraphics[width=0.75\linewidth]{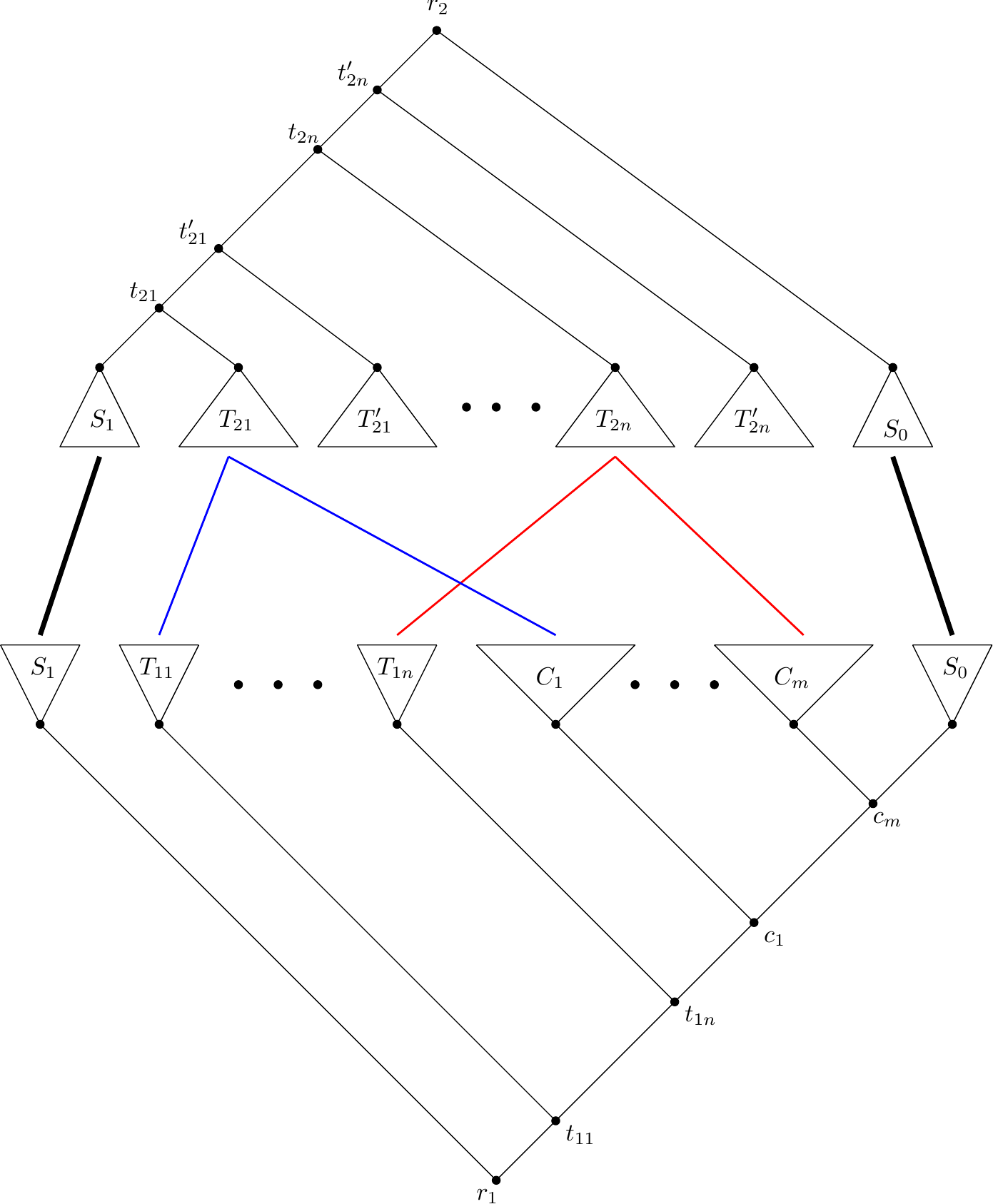}
  \caption{Trees $T_1$ with root $r_1$ and $T_2$ with root $r_2$ in instance
    $I$, obtained from $\tau(\phi)$. In an optimal solution trees $S_0$ and
    $S_1$ are fully aligned to each other (black lines).  If a variable
    gadget is in $M_l$ configuration (blue line), a clause in which the
    corresponding negative literal occurs can be matched optimally (blue
    line).  The same holds for configuration $M_{\bar{l}}$ and positive
    literal occurrences (red lines).}
  \label{fig:compl_constr}
\end{figure}

From variables gadgets (Figure~\ref{fig:var_gadget}) and clause gadgets
(Figure~\ref{fig:clause_gadget}) we construct two rooted trees $T_1$ and
$T_2$ as depicted in Figure~\ref{fig:compl_constr}, where trees $T_{2i}$ and
$\bar{T}_{2i}$ denote subtrees rooted at $v'$ and $v''$, respectively, in
$T_2$ (Figure~\ref{fig:var_gadget}). $T_1$ and $T_2$, together with cost
function \eqref{eq:exmpl_metric}, form our instance $\tau(\phi)$.  Both trees
connect subtrees of variable and clause gadgets in linear chains, augmented
by two separator trees $S_0$ and $S_1$. $S_1$ represents a complete binary
trees of depth $4$, and $S_0$ a complete binary tree of depth
$k=\lceil\log(40n^2+141)\rceil$.

We assign leaves of separator trees arbitrary but unique taxa in a way, such
that tree $S_i$ in $T_1$ is an identical copy of $S_i$ in $T_2$,
$i\in\{1,2\}$.

\begin{lemma} \label{theo:construct}
Consider the construction $\tau(\phi)$ in Figure \ref{fig:compl_constr}.  In
an optimal matching of trees $T_1$ and $T_2$, nodes in the \emph{backbone} of
$T_1$, $\mathcal{B}_1:=\{c_1,\dots,c_m, t_{11},\dots, t_{1n}\}$, and nodes in
the \emph{backbone} of $T_2$, $\mathcal{B}_2:=\{t_{22},\dots, t_{2n},
t'_{22},\dots, t'_{2n}\}$, are unmatched.
\end{lemma}

\begin{proof}
First, an optimal solution must match roots $r_1$, $r_2$, since edge
$(r_1,r_2)$ does not introduce any constraint on the remaining vertices and
has maximum label overlap.  Therefore, matching any node in $\mathcal{B}_1$
invalidates the matching of nodes in $S_0$.  According to conditions
(i)-(iii), a feasible matching cannot match nodes from different subtrees in
$\mathcal{T}_2=\{T_{21},\dots,T_{2n},T'_{21},\dots,T'_{2n},S_0,S_1\}$ to
nodes in $\mathcal{B}_1$.  Replacing all edges incident to nodes in
$\mathcal{B}_1$ by a full matching of nodes in $S_0$ reduces the cost by at
least
\begin{equation}
\begin{split}
 2\left(\sum_{u\in S_0} |Y(u)| -  \sum_{i=1}^n (|Y(t_{2i})| + |Y(t'_{2i})|) - \max_{T\in\mathcal{T}_2\setminus \{S_0\}}\sum_{v\in T} |Y(v)|\right)\\
\geq\; 2\left( (k+1)\cdot 2^k - 16 - 40n^2 - 125\right),
\end{split}
\end{equation}
where $k$ is the depth of $S_0$.
The upper bound of $125$ on $\sum_{v\in T_{2j}} |Y(u)|$ assumes that each variable occurs in at most $4$ clauses, and $125>\sum_{v\in S_1} |Y(u)|$.
Note that the taxa assigned to the 16 leaves of $S_1$ are contained only in $Y(r_1)$ and that for each $i$, $|Y(t_{2i})| + |Y(t'_{2i})|\leq 20$.
For the above chosen $k$ it holds $(k+1)\cdot 2^k > 40n^2+141$.

Similarly, a feasible matching cannot match nodes from different subtrees in
$\mathcal{T}_1:=\{C_1,\dots,C_m,T_{11},\dots,T_{1n},S_0,S_{1}\}$ to nodes in
$\mathcal{B}_2$.  Assume the optimal solution matches nodes in a subtree
$C_i$ to nodes in $\mathcal{B}_2$. Since every node in $\mathcal{B}_2$ is
ancestor of $S_1$, the nodes of $S_1$ are unmatched. Replacing the edges
between $C_i$ and $\mathcal{B}_2$ by a full matching of nodes in $S_1$
reduces the cost by at least
\[
  2(\sum_{u\in S_1} |Y(u)| - \sum_{v\in C_i} |Y(u)| )=2(80-59)>0,
\]
a contradiction.  An analog argument applies to matching nodes in one of the
trees $T_{1i}$ to nodes in $\mathcal{B}_2$, with $\sum_{u\in T_{1i}} Y(u) =
12 < \sum_{u\in S_1} Y(u)$.  As the optimal matching of trees $S_0$ has cost
0, matching at least one node in $S_0$ in $T_1$ to a node in $\mathcal{B}_2$
strictly increases the overall cost.
\end{proof}

Now we are ready to state the main theorem.

\begin{theorem}\label{theo:hard}
 For an instance of the \emph{minimum arboreal matching} problem with cost function \eqref{eq:exmpl_metric} and an integer $k$,
 it is NP-complete to decide whether there exists an arboreal matching of cost at most $k$.
\end{theorem}

\begin{proof}
 First, we show that if $\phi$ is satisfiable, then $\tau(\phi)$ admits a matching $M$ of cost $d(M_0)-10n-26m - 5\cdot 2^4 - (k+1)2^k - q$,  
 where $k$ is the depth of tree $S_0$ and $q$ is total number of leaves of tree $T_1$ or, equivalently, tree $T_2$. For this, let $\nu$ be a satisfying assignment for $\phi$. We start from $M=\emptyset$.
 For each variable $x_i$ we set the corresponding variable gadget to configuration $M_l$ if $\nu(x_i)=\mbox{false}$ and to configuration $M_{\bar{l}}$ if $\nu(x_i)=\mbox{true}$, each having cost
 $d(M_0)-10$ (Lemma~\ref{theo:var_gadget}). Additionally, we match each subtree representing a clause $C_j$ to subtree $T_{2i}$ or $T'_{2i}$ following the construction in Lemma~\ref{theo:clause_gadget},
 where literal $x_i$ or $\neg x_i$, respectively, is contained in $C_j$ and evaluates to true under the assignment $\nu$. Note that none of the ancestors of subtree $X_i$ or $\bar{X}_i$ 
 (see Figure~\ref{fig:var_gadget}), respectively,
 is matched in this case (Lemma~\ref{theo:var_gadget} and Lemma~\ref{theo:construct}). Each clause therefore contributes $d(M_0)-26$ to the overall cost
 (Lemma~\ref{theo:clause_gadget}). Finally, trees $S_0$ and $S_1$ are covered
 by full matchings of their nodes and the roots $r_1$, $r_2$ are matched,
 yielding a matching of total cost 
\begin{equation}\label{eq:optcost}
 d(M_0)-10n-26m - 5\cdot 2^4 - (k+1)2^k - q
\end{equation}
As an optimal solution matches roots $r_1$ and $r_2$ but none of the nodes in $\mathcal{B}_1$ or $\mathcal{B}_2$ (Lemma~\ref{theo:construct}), any optimal matching must match subtrees in 
 $\mathcal{T}_1:=\{C_1,\dots,C_m,T_{11},\dots,T_{1n},S_0,S_{1}\}$ and $\mathcal{T}_2=\{T_{21},\dots,T_{2n},T'_{21},\dots,T'_{2n},S_0,S_1\}$ optimally. Since an optimal matching of any tree in $\mathcal{T}_1$
 to $T_2$ and vice versa is given by Lemmas~\ref{theo:var_gadget} and~\ref{theo:clause_gadget}, one can always derive a satisfying assignment of $\phi$ from $M$. Therefore, if $\phi$ is not satisfiable, the
weight of a maximum matching in $\tau(\phi)$ is strictly larger than \eqref{eq:optcost}. 
\end{proof}


\section{An Integer Linear Program} \label{sec:ilp}

In this section we introduce a simple integer linear programming formulation for the problem of finding a minimum cost arboreal matching between $\mathcal{C}(T_1)$ and $\mathcal{C}(T_2)$, given
two rooted phylogenetic trees $T_1=(V_1,E_1)$, $T_2=(V_2,E_2)$, and a cost function $\delta$. We number clades $C$ in $\mathcal{C}(T_1)$ from $1$ to $|V_1|$ and clades $\bar{C}$ in $\mathcal{C}(T_1)$ from $1$ to $|V_2|$.
An indicator variable $x_{i,j}$ denotes whether $(C_i,\bar{C}_j)\in m$
($x_{i,j}=1$) or not ($x_{i,j}=0$). Set $\mathcal{I}$ contains pairs of
matched clades $\{(i,j), (k,l)\}$ that are \emph{incompatible} according to 
conditions (i)-(iii).
With $w(C_1,C_2):=\delta(C_1,-) +
\delta (-,C_2) - \delta(C_1,C_2)$ (see Section \ref{sec:arboreal_match}) a minimum cost arboreal matching is represented by the optimal solution to:
\begin{align}
\max\;&\sum_{i=1}^{|V_1|}\sum _{j=1}^{|V_2|} w(C_i,\bar{C}_j)x_{i,j}\\
\text{s.\,t.}\;& \sum_{j=1}^{|V_2|} x_{i,j}  \leq 1 & \forall i=1\dots |V_1|,\\
& \sum_{i=1}^{|V_1|} x_{i,j}  \leq 1 & \forall j=1\dots |V_2|,\\
 &x_{i,j}+x_{k,l}\leq 1 & \forall \{(i,j),(k,l)\} \in \mathcal{I},\\
& x_{i,j}\in\{0,1\}
\end{align}


\section{Evaluation}

We use a real-world dataset provided by Sul and Williams
\cite{sul08experimental} as part of the HashRF program.\footnote{Trees can be
  downloaded from \url{https://code.google.com/p/hashrf/}.}  It contains 1000
phylogenetic trees from a Bayesian analysis of 150 green
algae~\cite{lewis05unearthing}.  For the purpose of this comparison we
performed an all-against-all comparison of the first hundred trees in the
benchmark set as a proof-of-concept study, resulting in 5050 problem
instances.  We compute the values of the Robinson-Foulds metric as well as
the minimum arboreal matching using the Jaccard weights of order $k = 1$,
that is, the JRF metric $\dJRFone$. We limit the computation to two CPU
minutes per comparison and record the times for computing each value as well
as the best upper and lower bounds for $\dJRFone$.

\begin{figure}[tbp]
  \centering
  \includegraphics[width=\linewidth/2-6pt]{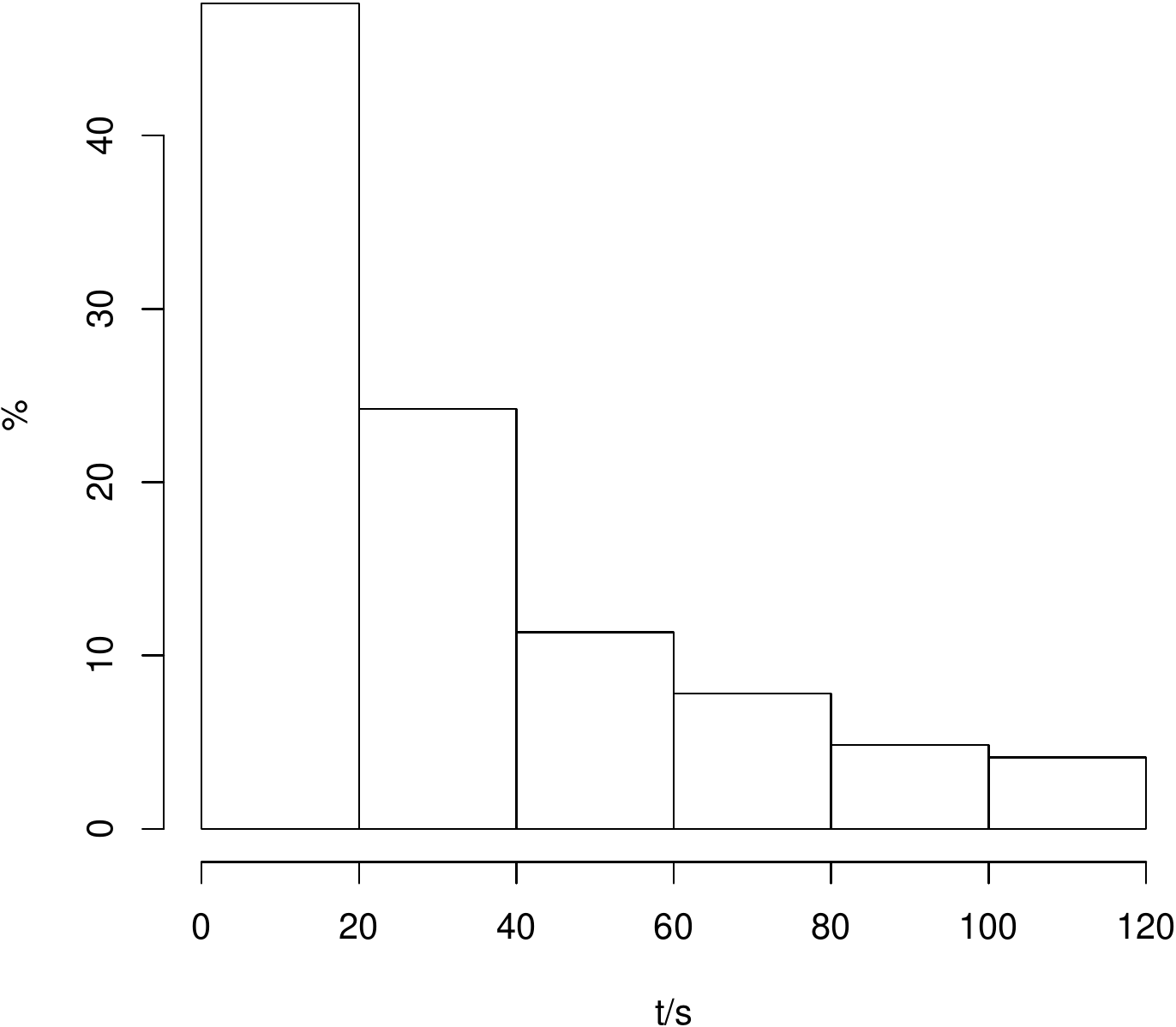}\hfill
 \includegraphics[width=\linewidth/2-6pt]{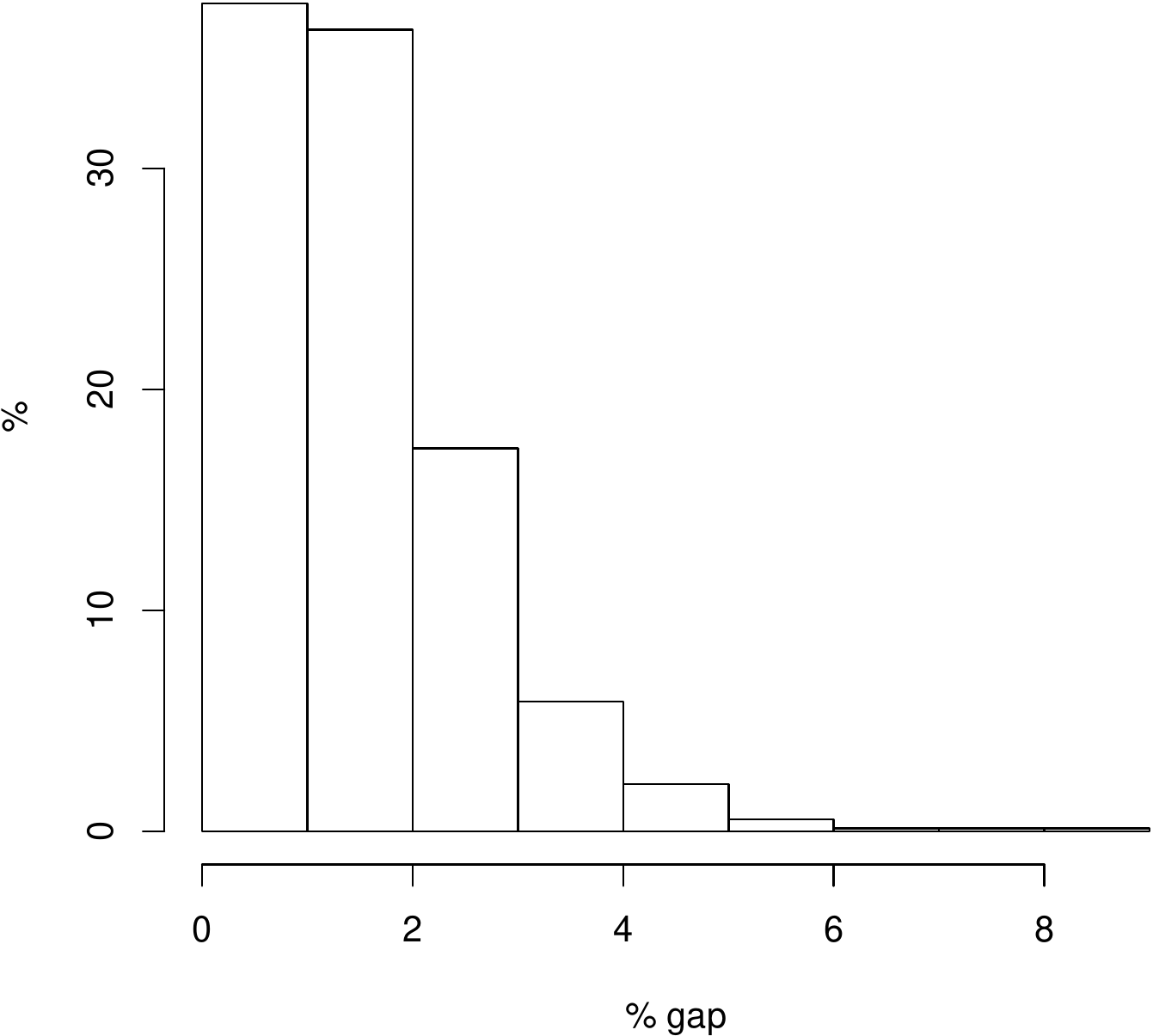}
 \caption{Running time and optimality gap statistics of the 5050
   benchmark instances. Left: histogram of running times of the 4300 instances
   that were solved to optimality within 2 CPU minutes. Right:
   histogram of the optimality gap in percent of the remaining 750 instances. This value is defined as $100
\cdot (u-l)/l$, where $u$ and $l$ are the upper and lower
bounds of the arboreal matching, respectively.}
\label{fig:histograms}
\end{figure}

From the 5050 instances, 4300 (85 \%) could be computed to optimality within
the time limit on an Intel Xeon CPU E5-2620 with 2.00 GHz\@. Most of these instances could
be solved within 40 CPU s. See Fig.~\ref{fig:histograms} (left) for a
histogram of running times. The remaining 750 instances (15 \%) were solved close
to optimality. Fig.~\ref{fig:histograms} (right) shows a histogram of
the relative optimality gap in percent. This value is defined as $100
\cdot (u-l)/l$, where $u$ and $l$ are the upper and lower
bounds of the arboreal matching, respectively. Overall, the majority
(3578 instances, 71 \%) could be solved to optimality within a
minute. Note that these results are obtained the quite
simple Integer Linear Programming formulation presented in this
paper. Improvements on the formulation will likely lead to a drastic
reduction of the running time.

Figure~\ref{fig:increasing_k} shows typical characteristics of
the arboreal JRF distances over increasing $k$ for a randomly picked
instance (tree 34 vs.\ tree 48). We observe that RF and $\dJRFk$
distances differ considerably for $k = 1$ and that $\dJRFk$ converges
quickly to RF\@ (Fig.~\ref{fig:increasing_k}, left). A similar converging behavior can be observed for the number of
matched clades (Fig.~\ref{fig:increasing_k}, right). The bottom plot
in Fig.~\ref{fig:increasing_k} illustrates the difference to
non-arboreal matchings. For $k=1$, the distances differ significantly
from the RF distance (25.8 versus 52), however, at the prize of a
large number of violations of the arboreal
property (91). As $k$ increases, the distance converges quickly to the
RF distance and the number of violations decreases. Note that zero
violations occur only when the non-arboreal distance is equal to the
RF distance. 

\begin{figure}[tbp]
  \centering
  \includegraphics[width=\linewidth/2-3pt]{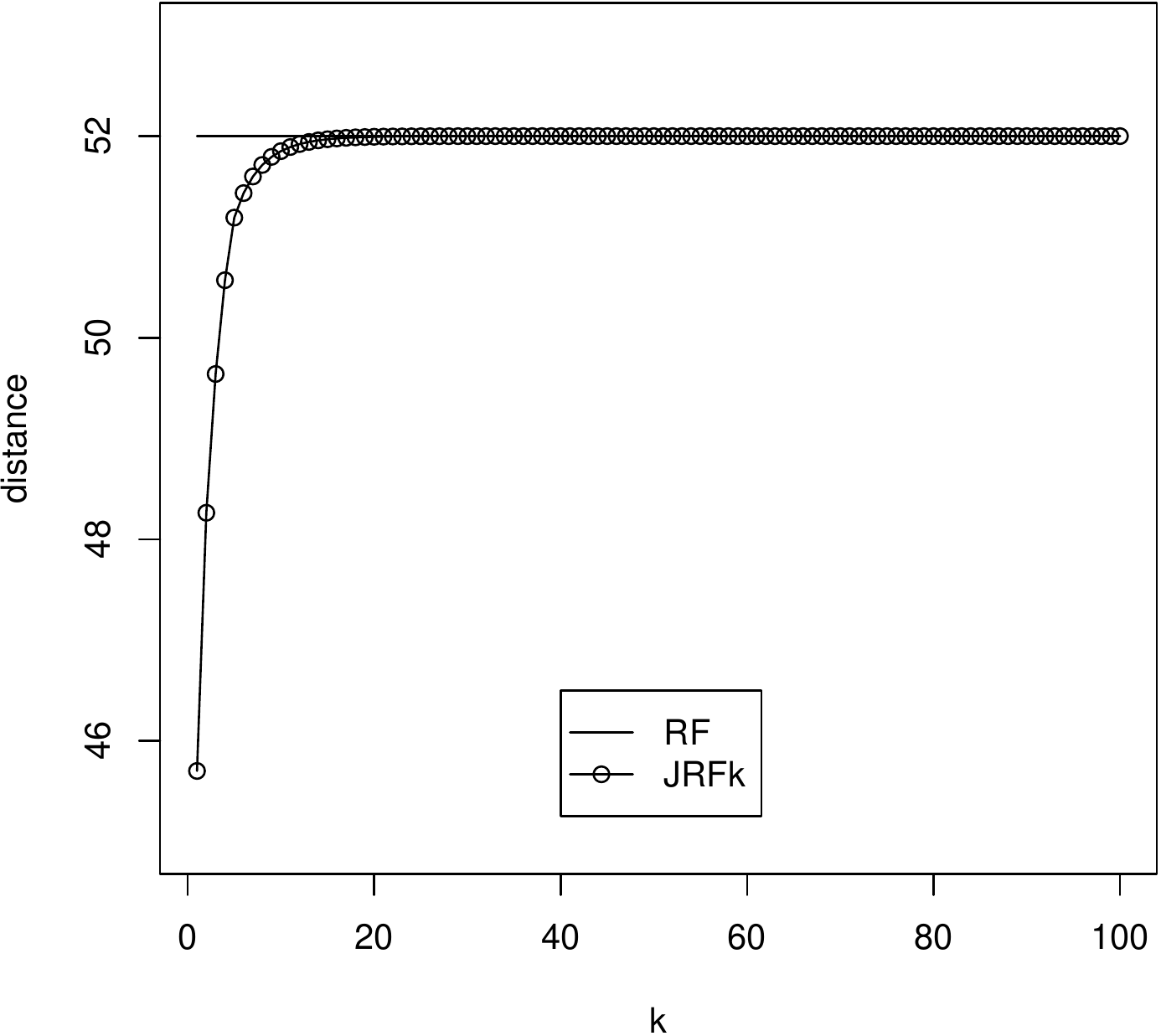}\hfill
  \includegraphics[width=\linewidth/2-3pt]{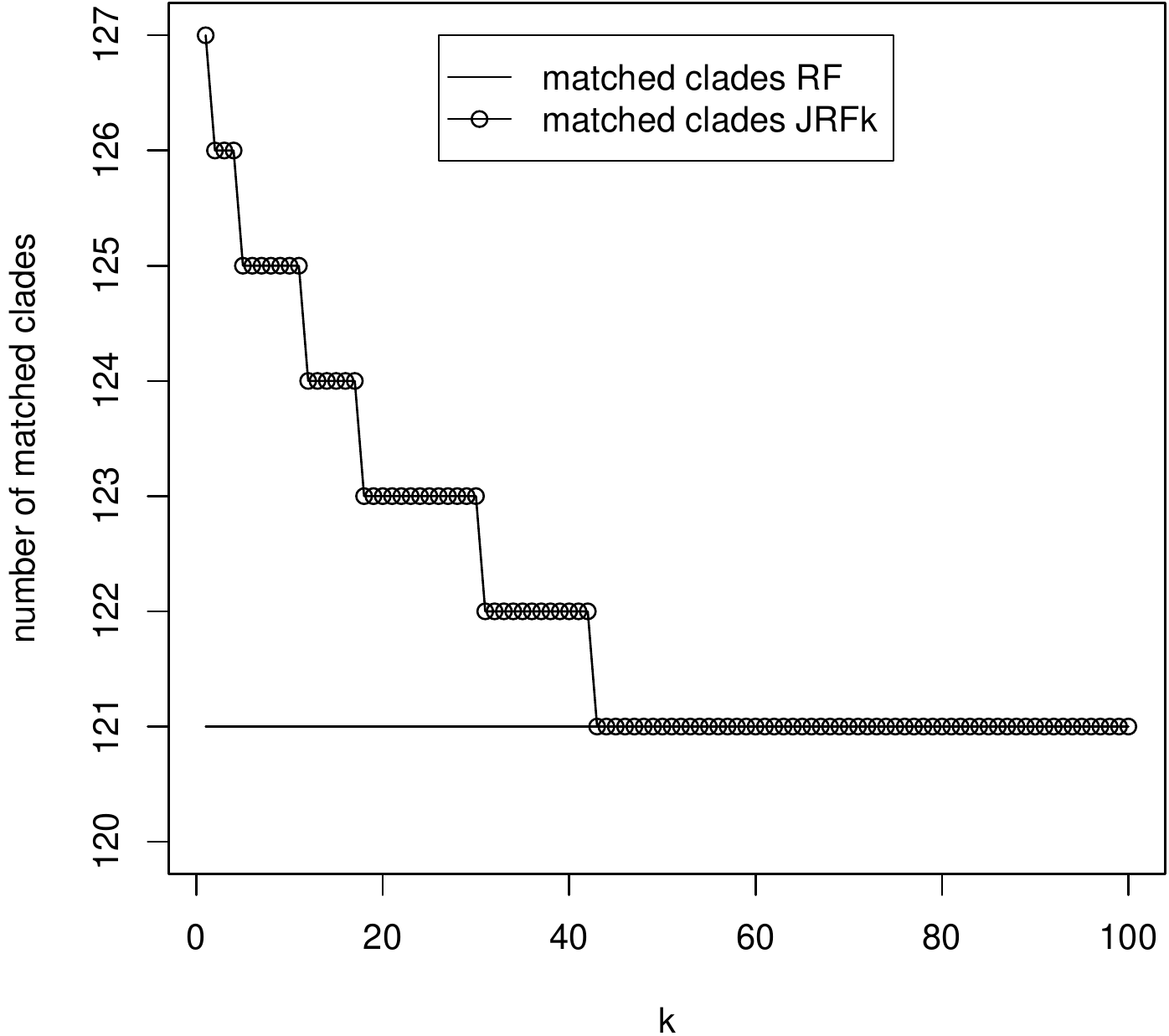}\\[1ex]
  \includegraphics[width=\linewidth/2+6pt]{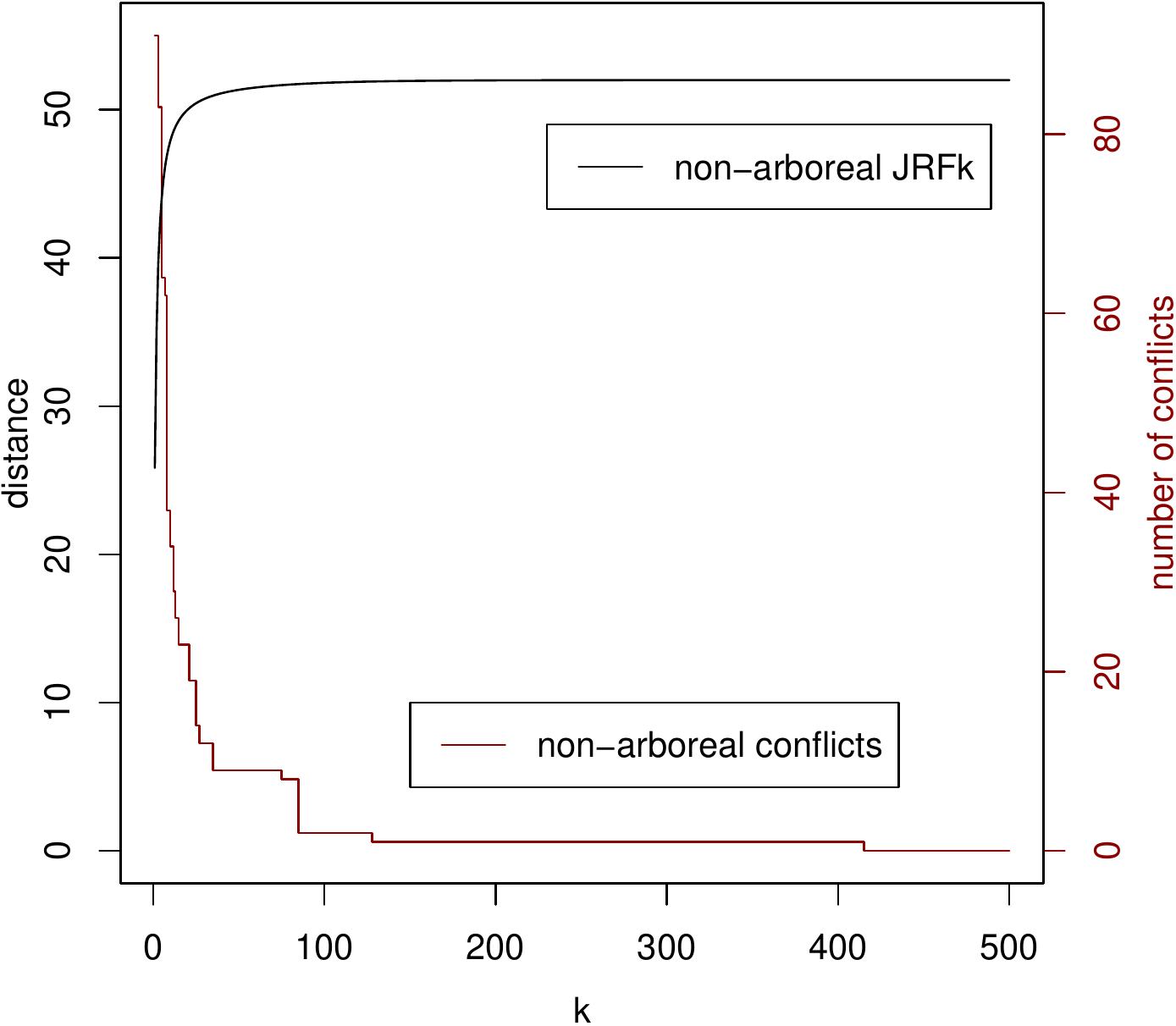}
  \caption{Typical characteristics of the distances over increasing
    $k$. In the randomly chosen example (tree 34 vs.\ tree 48), RF equals 52 and $\dJRFone$ is depicted by
    circles connected by lines (left plot). In the right plot we see
    that the number of matched clades decreases with increasing
    $k$. The plot below shows the development of distance and number
    of conflicts of the non-arboreal
    matching for increasing $k$.}
  \label{fig:increasing_k}
\end{figure}


\section{Conclusion}

We have introduced a tree metric that naturally extends the well-known
Robinson-Foulds metric.  Different from previous work, our metric is a true
generalization, as it respects the structure of the trees when comparing
clades.  Besides the theoretical amenities of such a generalization, our
methods naturally allows for a manual comparison of two trees, using the
arboreal matching that has been computed.  This allows us to compute ``best
corresponding nodes'' that respect the tree structures, and to inform the
user when other node correspondences disagree with the optimal matching.  We
believe that such a feature will be very useful for the manual comparison of two
trees, for example, in host-parasite comparison.

An open
question is the parameterized complexity of the problem, where natural
parameters are the size of the matching or, more relevant in applications,
the discrepancy between the size of the maximum arboreal matching and a
regular maximum matching.  The Maximum Independent Set problem is W[1]-hard
\cite{downey99parameterized} but, obviously, this does not imply that our
more restricted problem cannot be approached by a parameterized
algorithm~\cite{dabrowski11parameterized}.

We have come up with a
generalization that retains the advantages of the widely-used
Robinson-Foulds metric, but simultaneously overcomes
some of its shortcomings.  Our results are a first step to make the GRF and
JRF metrics applicable to practical problems.  In the future, faster
algorithms are needed for this purpose; we believe that such algorithms can
and will be developed.  Furthermore, we want to generalize our results for
unrooted trees, along the lines of~\cite[Sec.~2.1]{nye06novel}.  Here, the
main challenge lies in adapting the notion of an arboreal matching.

In the full version of this paper, we will evaluate the JRF metric following
ideas of Lin \etal~\cite{YuLin:2012dj}: That is, we will compare
distributions of distances with arboreal and non-arboreal matchings; and, we
will estimate the power of the new distance with regards to clustering
similar trees.


\paragraph{Acknowledgments.}

We thank W.\ T.\ J.\ White for helpful discussions.

\bibliographystyle{abbrv}
\bibliography{grf_wabi}

\end{document}